\documentclass[twoside,leqno,twocolumn]{article}

\usepackage[letterpaper]{geometry}

\usepackage{ltexpprt}
\usepackage{hyperref}

\usepackage{amsmath} 
\usepackage{amssymb}  
\usepackage[english]{babel}
\usepackage{xcolor}
\usepackage{graphicx}
\graphicspath{{./figures/}}

\newcommand{\ldef}{:=}
\newcommand{\rdef}{=:}
\newcommand{\Mcal}[1]{\mathcal{#1}}
\newcommand{\tth}{^{\text{th}}}
\newcommand{\bld}[1]{\mathbf{#1}}

\newcommand{\thmtitle}[1]{\mbox{}\textit{(#1).}}

\newcommand{\bulletsym}{\hbox{$\bullet$}}
\newcommand{\bulletend}{\relax\ifmmode\else\unskip\hfill\fi\bulletsym}
\newcommand{\squaresym}{\hbox{$\blacksquare$}}
\newcommand{\proofend}{\relax\ifmmode\else\unskip\hfill\fi\squaresym}
\newcommand{\trianglesym}{\hbox{$\blacktriangle$}}
\newcommand{\egend}{\relax\ifmmode\else\unskip\hfill\fi\trianglesym}

\renewenvironment{proof}{\textit{Proof.} }{\proofend}
\newtheorem{problem}{Problem}[section]

\newtheorem{Assumption}{Assumption}[section]

\newcommand{\until}[1]{\{1,\dots, #1\}}
\newcommand{\integernonnegative}{\ensuremath{\mathbb{Z}}_{\ge 0}}

\def \A{\bld{A}}
\def \Ahat{\widehat{\A}}
\def \abs{\mathrm{abs}}

\def \B{\bld{B}}
\def \Bhat{\widehat{\B}}

\def \E{\bld{E}}

\def \exp{\mathrm{exp}}
\def \F{\bld{F}}

\def \integer{\mathbb{Z}}

\def \M{\bld{M}}

\def \one{\bld{1}}

\def \real{\mathbb{R}}
\def \svec{\bld{s}}
\def \S{\bld{S}}
\def \smp{\mathcal{S}}
\def \smpind{s}

\def \u{\bld{u}}

\def \uhat{\widehat{\u}}
\def \v{\bld{v}}
\def \w{\bld{w}}
\def \wb{\overline{\bld{w}}}
\def \x{\bld{x}}

\def \xhat{\widehat{\x}}
\def \y{\bld{y}}
\def \z{\bld{z}}
\def \zero{\bld{0}}

\newcommand{\prob}[1]{\mathbb{P}\left(#1\right)}

\def \scriptA{\mathcal{A}}
\def \scriptB{\mathcal{B}}
\def \scriptxo{{\mathcal{X}_0}}
\def \scriptxf{{\mathcal{X}_f}}

\newcommand{\makesmall}[1]{\text{\small ${#1}$}}

\usepackage{algorithmic}

\newcommand\blfootnote[1]{%
  \begingroup
  \renewcommand\thefootnote{}\footnote{#1}%
  \addtocounter{footnote}{-1}%
  \endgroup
}

\newcommand{\mo}[1]{#1}
\newcommand{\nm}[1]{#1}

\begin{document}

\newcommand\relatedversion{}

\title{\Large Control of Discrete-Time LTI Systems using Stochastic Ensemble Systems\relatedversion}
\author{Nirabhra Mandal
\and Mohammad Khajenejad
\and Sonia Mart{\'\i}nez
\blfootnote{Nirabhra Mandal, Mohammad Khajenejad
and Sonia Mart{\'\i}nez are with the Mechanical and Aerospace Engineering Department, University of California San Diego. \texttt{\{nmandal,mkhajenejad,soniamd\}@ ucsd.edu}}
}

\date{}

\maketitle


\fancyfoot[R]{\scriptsize{Copyright \textcopyright\ 2023 by SIAM\\
Unauthorized reproduction of this article is prohibited}}





\begin{abstract} \small\baselineskip=9pt \textbf{
    In this paper, we study the control properties of a new class of
    stochastic ensemble systems that consists of families of random variables. These random variables provide an increasingly
    good approximation of an unknown discrete, linear-time invariant
    (DLTI) system, and can be obtained by a standard, data-driven
    procedure. Our first result relates the reachability properties
    of the stochastic ensemble system with that of the limiting DLTI
    system. We then provide a method to combine the control inputs
    obtained from the stochastic ensemble systems to compute a control
    input for the DLTI system. Later, we deal with a particular kind of
    stochastic ensemble systems generated from realizing Bernoulli
    random variables. For this, we characterize the variance of the computed
    state and control. We also
    do the same for a situation where the data is updated sequentially in a streaming fashion. We illustrate the results
    numerically in various simulation examples.}

\textbf{Keywords \textemdash 
  Approximate reachability, sample
  reachability, stochastic ensemble system.}
\end{abstract}

\section{Introduction}
\mo{

Ensemble control, \emph{i.e.}, investigating the ability of steering and manipulating an entire ensemble of (partially) unknown systems in a desired and optimal manner, has  
emerged from several science and engineering applications in recent years, \emph{e.g}, coordination of the movement of flocks in
biology \cite{brockett2010control}, manipulation of spin ensembles in nuclear magnetic
resonance \cite{li2006control}, \cite{glaser1998unitary}, or desynchronization of pathological neurons
in the brain in neuroscience~\cite{zhai2008control}. 

\paragraph{Literature Review} Motivated by this, a huge body of seminal work has been done
on the analysis of
(deterministic) ensemble controllability \cite{JSL-NK:07}, \cite{beauchard2010controllability}, \cite{AB-TB:12} and synthesis of optimal
ensemble controls \cite{JQ-AZ-JSL:13}, \cite{zlotnik2012synthesis}, \cite{wang2015fixed} by developing new analytical and numerical methods,
%
which have received
increasing attention due to their application to
robotics~\cite{AB-TB:12}, energy systems~\cite{MC-VYC-DD:18,
  AH-RM-DD-YD:20} and quantum control~\cite{JSL-NK:09,
  RB-NK:00}. However, the traditional ensemble control problem consists of driving a
collection of initial states of a continuum of systems to a set of final states with the same
control input~\cite{JSL-JQ:15}, which could be potentially conservative.  
The problem of reachability for finite-dimensional linear time-varying
systems has been studied in \cite{JSL:10, LT-JSL:16}, but here the system varies
in a continuum and not in a countable set. On the other hand, the
reachability properties of bilinear ensemble systems are studied
in~\cite{WZ-JSL:21}. In~\cite{JQ-AZ-JSL:13}, the authors explore the
problem of optimal control for stochastic linear systems. The
stochastic nature comes from additive noise to the
dynamics. Reference~\cite{MC-VYC-DD:18} solves an ensemble control
problem for devices with cyclic energy consumption patterns by
utilizing techniques from the Markov Decision Process framework.
Finally, the work~\cite{AH-RM-DD-YD:20} looks into a similar problem
but with uncertain dynamics and uses techniques from stochastic and
distributionally robust optimization in the process. However, in all the aforementioned work, either the control signal needs to be unique, or a continuum of ensemble (and partially unknown) systems is required to exist. This work aims to bridge this gap.

\paragraph{Contributions}
Unlike ensemble control, we let the system parameters
take values in a countable set. Moreover, we do not restrict the
control to be unique but allow every sample run of the system to
employ a different control function. This being said, we consider a so-called class of stochastic ensemble systems, which
arise from the approximation of systems whose exact parameters are
unknown. A known approach to obtain good parametric models consists of
learning the distribution of for the model whose realizations
correspond to a system approximation. As more data becomes available,
the distributions become more accurate, resulting in a stochastic
ensemble approximation. This enables us to i) study   
the reachability properties of a new class
of stochastic ensemble systems, ii) provide a procedure for combining
controls from such systems, iii) compute a control for the limiting DLTI
system, iv) characterize the variance of the state and control
of a stochastic ensemble system produced by means of a Bernoulli
distribution, and v) derive an improved result by means of least squares error
minimization. 
 As conclusions of these main contributions, we also can 
 vi) compute the variance of the state and control of a
 stochastic ensemble system that is obtained in an streaming
 fashion. 
 We illustrate our results in numerical examples.}
%
%
%
%

\paragraph{Notations}

We denote the set of real numbers, the set of non-negative real
numbers, the set of integers, and the set of non-negative integers
using $\real$, $\real_{\geq 0}$, $\integer$, and $\integernonnegative$,
respectively. We let $\real^n$ (similarly $\integernonnegative^n$) be the
Cartesian product of $\real$ (similarly $\integernonnegative$) with itself. $\real^{n \times m}$ denotes the set of real matrices of order
$n \times m$. 
The $i\tth$ component of
a vector $\v \in \real^n$ is denoted by $\v_i$ and the $ij\tth$ entry
of a matrix $\M \in \real^{n \times m}$ is denoted by $\M_{ij}$. For a
set of square matrices $\{\M^{(s)}\}_{i \in \integer}$, the product
operator is used to denote multiplication from right to left,
\emph{i.e.} $\prod_{i=1}^n \M^{(s)} \ldef \M^{(n)}\cdots\M^{(1)} $. For
a vector $\v \in \real^n$, $\v \geq 0$ denotes a component-wise
inequality.  We let $\abs :\real^{{n}} \to \real^{{n}}$ be a function
that produces a vector with component-wise absolute values from the
input vector, \emph{i.e.} $\abs(\v)_i = |\v_i| \in \real$,
$\forall \, i \in \until{n}$.
We denote the empty set using $\varnothing$.
\section{Problem Formulation}
Consider a discrete-time, unknown, linear time-invariant (LTI) system of the form

\vspace{-.4cm}
{\small
\begin{equation}
	\x(k+1) = \A\x(k) + \B\u(k),
	\label{eq:dlti}
\end{equation}}
where $\x \in \real^{{n}}$ is the state, $\x(0) = \x_0$ and $\x_f$ are the initial and desired final states, respectively, $\u \in \real^{{m}}$ is the control input, $\A \in \real^{{n}
  \times {n}}$ and $\B \in \real^{{n} \times {m}}$. 
The system matrices $\A$ and $\B$, as well as the initial state 
$\x_0$ are unknown, but it is possible to construct
increasingly good approximations of these system parameters using some
known approach, e.g., any data-driven or Machine Learning-based
method. This approximation is assumed to be done through the following newly defined class of 
\textit{stochastic ensemble systems}, for which we aim to characterize its reachability properties. 

\begin{Definition}\thmtitle{Stochastic ensemble system}
  Consider sequences of
  independent random variables
  $\scriptxo \ldef\{\x^{(s)}_0 \in \real^n \}_{s \in
    \integernonnegative}$,
  $\scriptA \ldef\{\A^{(s)} \in \real^{n \times n}\}_{s \in
    \integernonnegative}$ and
  $\scriptB \ldef\{\B^{(s)} \in \real^{n \times m}\}_{s \in
    \integernonnegative}$
such that $\scriptA, \scriptB,$ and $\scriptxo$ are all
independent of each other. For every $\sigma \in \integernonnegative$, a stochastic ensemble system is a linear, time
varying (LTV) system of the form: 

\vspace{-.4cm}
{\small
\begin{subequations}
\begin{align}
	\label{eq:sparse_dyn}  \xhat(\sigma,k\hspace{-.1cm}+\hspace{-.1cm}1) &\hspace{-.1cm}=\hspace{-.1cm} \Ahat(s_A(\sigma,k))\,\xhat(\sigma,k)
	 \hspace{-.1cm}+\hspace{-.1cm} \Bhat(\smpind_B(\sigma,k))\,\uhat(\sigma,k) \\
	\label{eq:sparse_init}  \xhat(\sigma,0) &\hspace{-.1cm}= \x_0^{(\smpind_0(\sigma))} \in \Mcal{X}_0\,,	
\end{align}
\label{eq:sparse_ltv_dyn}
\end{subequations}}
where, the \mo{indicator} functions
   $s_A:\integernonnegative \times \integernonnegative \to \integernonnegative, \quad s_B:\integernonnegative \times \integernonnegative \to \integernonnegative$ and
   $s_0:\integernonnegative \to \integernonnegative$ 
are 
used to index different realizations of the random variables in $\Mcal{X}_0$, $\Mcal{A}$, and $\Mcal{B}$, respectively. Thus, $\Ahat (s_A(\sigma, k))$ (similarly $\Bhat(s_B(\sigma, k))$) represents a realization of
  $\A^{(s_A(\sigma, k))} \in \scriptA$ (respectively $\B^{(s_B(\sigma, k))} \in \scriptB$) for each $\sigma$ and $k$. 
  \bulletend
  \end{Definition}
  We use the terms random variable and realization interchangeably.
 The $\sigma$ parameter is used to distinguish between different runs of the
system. Also note that since the initial state and the system matrices in \eqref{eq:sparse_ltv_dyn} are realizations of different random variables, each state $\xhat(\sigma,k)$ is also a realization of some random variable.
Moreover, we need to formally define the notion of \emph{stochastic ensemble approximation}. 
\begin{Definition} \label{def:sparse_approx} \thmtitle{Stochastic
    ensemble approximation} 
    Let $\M \in \real^{n \times m}$. Consider a sequence of independent
  random variables
  $\mathcal{M}\ldef\{\M^{(s)}\}_{s \in \integernonnegative}$ such that
  $\exists B \in \real$ such that $\|\M^{(s)}\| \leq B$, 
  $\forall s \in
  \integernonnegative$. 
  Then, $\Mcal{M}$ is a stochastic ensemble approximation of $\M$ if
  and only if 
  
  \vspace{-.4cm}
  {\small
  \begin{align*}
  \forall  \varepsilon > 0, \ \mathbb{P}({\|\M - \textstyle{\sum}_{s = 1}^d w_s \M^{(s)} \|_p > \varepsilon }) \to 0 \,\,\, \mathrm{as} \, d \to \infty,
  \end{align*}}
  with respect to some $p-$norm $\| \cdot \|_p$
and for some $w_s$'s of the form
  $[w_1,\cdots,w_d]^\top \in \smp^{d} \ldef \{\y \in \real^d \,|\, \one^\top \y = 1, \y \geq 0\}$ for each $d \in \integernonnegative$.
\bulletend
\end{Definition}
Furthermore, we assume the following.
\begin{Assumption} \label{ass:stoch_samp}
There exists a stochastic ensemble system in the form of \eqref{eq:sparse_ltv_dyn} that
consists of \emph{stochastic ensemble approximations} of the original system \eqref{eq:dlti} as defined below. 
\end{Assumption}
Being concerned with 
meeting certain reachability requirements for \eqref{eq:sparse_ltv_dyn}, next we provide a formal definition for the notion of approximate reachability. 

\begin{Definition} \label{def:approx_reach}
    \thmtitle{Approximate reachability of final state from an initial
      state} Let $\mathcal{X}_0$, $\mathcal{A}$, and $\mathcal{B}$ be
    stochastic ensemble approximations of $\x_0$, $\A$, and $\B$
    respectively.  Consider the stochastic ensemble system
    in~\eqref{eq:sparse_ltv_dyn} and let $\x_f \in \real^n$. Suppose
    there exists a fixed time $K \in \integernonnegative$ and for any realization of the random variables $\xhat(\sigma,0), \Ahat(s_A(\sigma,k)),$ and $\Bhat(s_B(\sigma,k))$ 
    there exist controls
    $\uhat(\sigma,0),\cdots,\uhat(\sigma,K-1)$ that drive the realization $\xhat(\sigma,0)$ through the trajectory
    $\xhat(\sigma,k)$, $\forall k \in \until{K}$ under the
    dynamics~\eqref{eq:sparse_ltv_dyn}, \emph{i.e.}
    
   \vspace{-.4cm}
    {\small
\begin{align}
  \notag \makesmall{ \xhat(\sigma,K)} & = \makesmall{\textstyle{\prod}_{k=0}^K \Ahat(s_A(\sigma,k)) \,
                         \xhat(\sigma,0)}
  \\
                       & \makesmall{+ \textstyle{\sum}_{k=0}^{K-1}
                         [ \textstyle{\prod}_{i=k+1}^{K} \Ahat(s_A(\sigma,i)) ] \Bhat(s_B(\sigma,k)) \uhat(\sigma,k)}  \,.
	\label{eq:sparse_control_necessary}
\end{align}}
Then $\x_f$ is said to be approximately reachable from $\x_0$ if and only if $\{\xhat(\sigma,K)\}_{\sigma \in \integernonnegative}$ is a stochastic ensemble approximation of $\x_f$.
\bulletend
\end{Definition}

Further, we need to formally introduce the notion of sample reachability as follows.
\begin{Definition} \label{def:sample_reach} \thmtitle{Sample reachability of final state from an initial state in unit time} 
Let $\mathcal{X}_0$, $\mathcal{X}_f$, and $\mathcal{A}$ be stochastic ensemble approximations of $\x_0$, $\x_f$ and $\A$ respectively, \mo{and $s_f : \integernonnegative \to \integernonnegative$ be an indicator function  that denotes the realizations of the random variables in the set $\scriptxf \ldef \{\x_f^{(s)} \in \real^n\}_{s \in \integernonnegative}$}. 
Consider the stochastic ensemble system in~\eqref{eq:sparse_ltv_dyn} with $\Bhat(k) = \B$.
Suppose there exists a control $\{\uhat(\sigma,0)\}_{\sigma \in \integernonnegative}$ for each $\sigma \in \integernonnegative$ that drives $\xhat(\sigma,0) = \x_0^{(s_0(\sigma))}$ to $\xhat(\sigma,K) = \x_f^{(s_f(\sigma))}$ under the dynamics~\eqref{eq:sparse_ltv_dyn}, in one ($K = 1$) time step, \emph{i.e.} 

\vspace{-.4cm}
{\small
\begin{align}
	\x_f^{(s_f(\sigma))} & = \Ahat(s_A(\sigma,1)) \,\, \x_0^{(s_0(\sigma))} + \B \, \uhat(\sigma,0)  \,.
	\label{eq:sparse_control}
\end{align}}
Then $\x_f$ is  sample reachable from $\x_0$ in unit time.
\bulletend
\end{Definition}
Note that the control $\uhat(\sigma,0)$ is chosen in order to drive $\x_0^{(s_0(\sigma))}$ to $\x_f^{(s_f(\sigma))}$ in the previous definition.

Our problem of interest can cast as follows. 
\begin{problem}
Given the aforementioned setup and Assumption \ref{ass:stoch_samp}, characterize the relation between approximate reachability under the stochastic ensemble system \eqref{eq:sparse_ltv_dyn} and the standard  reachability under the LTI system \eqref{eq:dlti}. Moreover, find controls that drive $\x_0$ to $\x_f$ under~\eqref{eq:dlti} using the stochastic ensemble system~\eqref{eq:sparse_ltv_dyn}. 
\end{problem}
We conclude this section by highlighting the difference between our notion of stochastic ensemble system with that of (deterministic) ensemble control.

\begin{Remark}\label{rem:ensemble_control}
\thmtitle{Comparison with ensemble control} {\rm The stochastic
  ensemble system in~\eqref{eq:sparse_ltv_dyn} has a similar structure
  to ensemble systems~\cite{JSL-NK:07}. However, the $\sigma$ that
  parameterizes the system takes discrete values. Moreover, $\Ahat$ and $\Bhat$
  do not vary continuously with respect to $\sigma$. We also allow the controls
  for each $\sigma$ to be different, i.e., we do not seek one single control
  that steers the ensemble system between points of interest. Finally, if a system is ensemble controllable, then similar results to
  the ones in this paper can be given, generalizing
  ensemble control systems and their objective.}  \bulletend
\end{Remark}
\vspace{-.1cm}


\section{Control using Stochastic Ensemble Systems}

In this section, we first show that reachability under the DLTI system~\eqref{eq:dlti} is sufficient for approximate reachability\footnote{Not to be misinterpreted as the reachability of autonomous systems, studied e.g., in the authors' previous works \cite{khajenejad2023tight,DirectPol22}.} under the stochastic ensemble systems. 

Recall that we use $\sigma \in \integernonnegative$ to index the runs of the system.
Moreover, $s_0(\sigma)$ (respectively $s_A(\sigma,k)$, $s_B(\sigma,k)$) indexes the realization of the initial state (respectively of matrices $\A$ and $\B$ at the $k\tth$ time step) used in~\eqref{eq:sparse_ltv_dyn}. In the sequel, we omit the $\sigma$ for the sake of brevity wherever there is no confusion.
Now we are ready to state the result.

\begin{lemma}\label{lem:sparse_control_necessary}
\thmtitle{Sufficient condition for approximate reachability using state reachability} 
Let $\mathcal{X}_0$, $\mathcal{A}$, and $\mathcal{B}$ be stochastic ensemble approximations of $\x_0$, $\A$, and $\B$ respectively. Suppose that $\x_f$ is reachable from $\x_0$ in $K \in \integernonnegative$ time steps under the dynamics~\eqref{eq:dlti}.  Then, $\x_f$ is approximately reachable from $\x_0$ under the stochastic ensemble system in~\eqref{eq:sparse_ltv_dyn}.
\end{lemma}

\begin{proof}
Since $\scriptxo, \scriptA,$ and $\scriptB$ are stochastic ensemble approximations of $\x_0$, $\A$, and $\B$ respectively $\exists B_\scriptxo, B_\scriptA$, and $B_\scriptB$ such that $\|\x_0^{(s)}\| \leq B_\scriptxo$, $\|A^{(s)}\| \leq B_\scriptA$, and $\|B^{(s)}\| \leq B_\scriptB$, $\forall s \in \integernonnegative$.
  Now, by the hypothesis, there exists controls $\u(0), \cdots, \u(K-1)$
  that drive $\x_0$ to $\x_f$ under the
  dynamics~\eqref{eq:dlti}. Moreover since $K$ is finite, all the controls can be bounded by a constant $B_\u$.
    We show that
  $\uhat(\sigma,k) = \u(k)$, $\forall k \in \until{K-1}$ 
proves the claim.
Let $\varepsilon > 0$. If we use the controls $\u(0), \cdots, \u(K-1)$ in~\eqref{eq:sparse_control_necessary}, we obtain
\mo{
\begin{align*}
\begin{array}{c}
	\makesmall{\y(1) \hspace{-.1cm}- \hspace{-.1cm} \sum_{s_0 = 1}^N w_{s_0,1} \xhat(s(1))   \hspace{-.1cm}=  \hspace{-.1cm}\Ahat(s_A(1)) ( \x_0 \hspace{-.1cm} - \hspace{-.1cm} \sum_{s_0 = 1}^N w_{s_0,1} \x_0^{(s_0)} )}  \\
	 \makesmall{+ \Bhat(s_B(1)) ( \u(0) - \sum_{s_0 = 1}^N w_{s_0,1} \u(0) ),}
	 \end{array}
\end{align*}}
at $k=1$ starting at $\x_0^{(s_0)}$ and for $[w_{s_0,1}] = [w_{1,1},\cdots,w_{N,1}] \in \smp^N$ which makes $\scriptxo$ a stochastic ensemble approximation of $\x_0$ and $\y(1) = \Ahat(s_A(\sigma,1)) \x_0 + \Bhat(s_B(\sigma)) \u(0)$.
Note that here we are using the weights that make $\Mcal{X}_0$ a stochastic ensemble approximation throughout and we use the independence properties of the random variables to separate the product and the sum.  
Now, since $\sum_{s_0 = 1}^N w_{s_0,1} \u(0) = \u(0)$, then

\vspace{-.4cm}
{\small
\mo{
\begin{align}
	 & {\| \y(\hspace{-.05cm}1\hspace{-.05cm}) \hspace{-.1cm}-\hspace{-.1cm} \textstyle{\sum}_{s_0 = 1}^N \hspace{-.05cm}w_{s_0,1} \xhat(s(\hspace{-.05cm}1\hspace{-.05cm})\hspace{-.05cm})\hspace{-.05cm} \|
	 \hspace{-.1cm} \leq\hspace{-.1cm} B_\scriptA \| \x_0 \hspace{-.1cm}- \hspace{-.1cm}\textstyle{\sum}_{s_0 = 1}^N \hspace{-.05cm}w_{s_0,1} \x_0^{(s_0)}\hspace{-.05cm} \| }.
	  \label{eq:for_tr_ineq1}
\end{align}
}}
Next, by defining $\z(1) = \Ahat(s_A(1)) \x_0 + \B \u(0)$
and considering a $[w_{s_B(1),2}] = [w_{1,2},\cdots,w_{N,2}] \in \smp^N$, $\scriptB$ becomes a stochastic ensemble approximation of $\B$. Then, by similar arguments as before
\vspace{-.2cm}
{\small
\mo{
\begin{align}\label{eq:for_tr_ineq2}
\begin{array}{rl}
	  &\| \z(1) - \textstyle{\sum}_{s_B(1) = 1}^N w_{s_B(1),2} \, \y(1) \| 
	 \hspace{-.1cm} \leq \\
	  & B_\u \| \B - \textstyle{\sum}_{s_B(1) = 1}^N w_{s_B(1),2} \B(s_B(1)) \|.	 
	  \end{array}
\end{align}
}}
Finally, for $\x(1) = \A \x_0 + \B \u(0)$, we have 
\vspace{-.2cm}
\mo{
\begin{align}\label{eq:for_tr_ineq3}
\begin{array}{rl}
	 & \makesmall{ \| \x(1) - \textstyle{\sum}_{s_A(1) = 1}^N w_{s_A(1),3} \z(1) \|} \\
	 & \leq \makesmall{ \| \A - \textstyle{\sum}_{s_A(1) = 1}^N w_{s_A(1),3} \Ahat(s_A(1)) \| \|\x(0)\|},
	 \end{array}
\end{align}
}
using $[w_{s_A(1),2}] = [w_{1,2},\cdots,w_{N,2}] \in \smp^N$ which makes $\scriptA$ a stochastic ensemble approximation of $\A$.
Then, using the triangle inequality of the norms and~\eqref{eq:for_tr_ineq1}--\eqref{eq:for_tr_ineq3}, we \mo{obtain}
\vspace{-.2cm}
\mo{
\begin{align*}
\begin{array}{rll}
	\hspace{-.3cm}& \makesmall{\| \x(1) - \textstyle{\sum}_{s_0 = 1}^N w_{s_0,1} \xhat(\sigma,1) \| \leq}\\
	\hspace{-.3cm}& \makesmall{ \alpha [ \| \x_0 \hspace{-.1cm}- \hspace{-.1cm}\sum_{s_0 = 1}^N w_{s_0,1} \x_0^{(s_0)} \| \hspace{-.1cm}+ \hspace{-.1cm}\| \B \hspace{-.1cm}- \hspace{-.1cm}\sum_{s_B(1) = 1}^N \hspace{-.05cm}w_{s_B(1),2} \B(s_B(1)\hspace{-.05cm}) \|}\\
	\hspace{-.3cm}& \makesmall{+ \| \A - \textstyle{\sum}_{s_A(1) = 1}^N w_{s_A(1),3} \Ahat(s_A(1)) \| ]},
	\end{array}
\end{align*}
}
where $\alpha = \max\{B_\scriptA,B_\u,\|\x(0)\|\}$. Hence,
\vspace{-.2cm}
\mo{
\begin{align*}
\begin{array}{rlll}
	 & \makesmall{\prob{\| \x(1) - \textstyle{\sum}_{s_0 = 1}^N w_{s_0,1} \xhat(\sigma,1) \| < \varepsilon}) \geq} \\
	 & \makesmall{\mathbb{P}({\| \A - \textstyle{\sum}_{s_A(1) = 1}^N w_{s_A(1),3} \Ahat(s_A(1)) \|  < \frac{\varepsilon}{3\alpha}}) +} \\
	 & \makesmall{\mathbb{P}({\| \B - \textstyle{\sum}_{s_B(1) = 1}^N w_{s_B(1),2} \B(s_B(1)) \| < \frac{\varepsilon}{3\alpha}}) +} \\
	 & \makesmall{\mathbb{P}({\| \x_0 - \textstyle{\sum}_{s_0 = 1}^N w_{s_0,1} \x_0^{(s_0)} \| < \frac{\varepsilon}{3\alpha}}) \to 1 \,\,\mathrm{as}\,\, N \to \infty}.
\end{array}
\end{align*}
}
Thus, $\{\xhat(s(1))\}_{s(1) \in \integernonnegative}$ forms a stochastic ensemble approximation of $\x(1)$. 
Thus, using the fact that $\{\x_0^{(s_0)}\}_{s_0 \in
  \integernonnegative}$, $\{\Ahat(s_A(1))\}_{s_A(1) \in
  \integernonnegative}$, and $\{\Bhat(s_B(1))\}_{s_B(1) \in
  \integernonnegative}$ are stochastic ensemble approximations of
$\x_0$, $\A$, and $\B$ respectively, we have shown that $\{\xhat(s(1))\}_{s(1) \in
  \integernonnegative}$ forms a stochastic ensemble approximation of
$\x(1)$. Next using the fact that
$\{\xhat(s(1))\}_{s(1) \in \integernonnegative}$,
$\{\Ahat(s_A(2))\}_{s_A(2) \in \integernonnegative}$, and
$\{\Bhat(s_B(2))\}_{s_B(2) \in \integernonnegative}$ are stochastic
ensemble approximations of $\x(1)$, $\A$, and $\B$ respectively we can
show that $\{\xhat(s(2))\}_{s(2) \in \integernonnegative}$ forms a
stochastic ensemble approximation of $\x(2)$. By induction on
$\x(1),\cdots,\x(K)$, the proof can be completed.
\end{proof}

Note that the previous lemma states that the desired state $\x_f$ is reachable from $\x_0$ only if $\x_f$ is approximately reachable from $\x_0$. Also, note that we deal with a much weaker version of reachability here. All we require is that $\x_f$ be in the reachable subspace from $\x_0$. We do not require the whole $\real^{{n}}$ to be reachable.
  
Next, \mo{by using the notion of sample reachability (Definition~\ref{def:sample_reach}), we synthesize a control input sequence that drives the system in~\eqref{eq:dlti} from $\x_0$ to $\x_f$.} 
 
\begin{lemma}\label{lem:single_step_ltv}
\thmtitle{Approximation of control of LTI systems using stochastic ensemble systems}
Consider the system in~\eqref{eq:dlti} with initial state $\x_0$ and desired final state $\x_f$. Let $\mathcal{X}_0$, $\mathcal{X}_f$ and $\mathcal{A}$ be stochastic ensemble approximations of $\x_0$, $\x_f$ and $\A$ respectively. Next consider the stochastic ensemble system in~\eqref{eq:sparse_ltv_dyn} with $\Bhat(k) = \B$, $\forall k \in \until{K-1}$.
Suppose that $\x_f$ is sample reachable from $\x_0$ in unit time. Let $\{\uhat(\sigma,0)\}_{\sigma \in \integernonnegative}$ be as in Definition~\ref{def:sample_reach}.
Then for each $\varepsilon > 0$,
{\small
\begin{align}
	\lim_{N\to \infty} \mathbb{P}({\|\B\v - \B \textstyle{\sum}_{\sigma = 1}^N w_{\sigma} \uhat(\sigma,0)\| > \varepsilon}) \to 0
\end{align}}
with $w_\sigma$'s of the form $[w_1,\cdots,w_N]^\top \in \smp^{N}$ for each $N \in \integernonnegative$ and with $\B\v$ satisfying
\begin{align}
	\x_f =  \A \x_0 + \B \, \v \,.
	\label{eq:actual_control}
\end{align}
\end{lemma}

\begin{proof}
We do this in a very similar way to the proof of Lemma~\ref{lem:sparse_control_necessary}. So we skip most of the details due to lack of space. Consider an $ \varepsilon > 0$. Then, consider the state evolution equations
{\small
\begin{subequations}
\begin{align}
	\label{eq:sparse_control_step1} \x_f^{(s_f)} & = \Ahat(s_A(1)) \x_0 + \B \u'(0)  \,, \\
	\label{eq:sparse_control_step3} \x_f & =  \Ahat(s_A(1)) \x_0 + \B \u''(0) \,.
\end{align}
\end{subequations}}
Note that~\eqref{eq:sparse_control_step1} has been formulated by choosing controls such that they drive the initial condition $\x_0$ (and not the realizations) to the final state $\x_f^{(s_f)}$. Similarly~\eqref{eq:sparse_control_step3} has been formulated by choosing controls such that they drive the initial condition $\x_0$ to the final state $\x_f$ (and not the realizations). Then, using very similar arguments as in the proof of Lemma~\ref{lem:sparse_control_necessary} it is possible to show that the controls $\u',$ and $\u''$ exist and can be attained by taking the weighted sum of the controls in~\eqref{eq:sparse_control}. 
%

Moreover, using the triangle inequality of the norms and probability theory, it can be shown that

\vspace{-.4cm}
{\small
\mo{
\begin{align*}
	& \makesmall{\mathbb{P}({\| \B\v - \B \textstyle{\sum}_{s_0 = 1}^N w_{s_0,1} \uhat(s_0)\| < \varepsilon}) \geq} \\
	 & \makesmall{\mathbb{P}({\| \A - \textstyle{\sum}_{s_A(1) = 1}^N w_{s_A(1),3} \Ahat(s_A(1)) \| < \frac{\varepsilon}{3\beta}}) +} \\
	 & \makesmall{\mathbb{P}({\| \x_f - \textstyle{\sum}_{s_f = 1}^N w_{s_f,2} \x_f^{(s_f)} \|  < \frac{\varepsilon}{3\beta} }) +} \\
	 & \makesmall{\mathbb{P}({\| \x_0 - \textstyle{\sum}_{s_0 = 1}^N w_{s_0,1} \x_0^{(s_0)} \| < \frac{\varepsilon}{3\beta}}) \to 1 \,\,\mathrm{as}\,\, N \to \infty}\,,
\end{align*}}}
where 
	$\beta = \max\{B_\scriptA,1,\|\x_0\|\}$.
\end{proof}

The previous result states that if $\x_f$ is sample reachable from $\x_0$ in unit time, $\x_f$ is also reachable under~\eqref{eq:dlti}. It is worthwhile to note that we are not approximating $\B$ using a stochastic ensemble approximation. If we did, then essentially we would introduce convolution-like sums which makes the problem more complicated. Further, if $\x_f$ is reachable from $\x_0$ only in multiple time steps, then the matter of non-unique control sequences poses another problem. 
%

Note that the claim in Lemma \ref{lem:single_step_ltv} works for any sequence of random variables that form stochastic ensemble approximations as defined in Definition~\ref{def:sparse_approx}. In the next section we study stochastic ensemble systems that are generated from sampling the DLTI system parameters in~\eqref{eq:dlti}.

\section{On Stochastic Ensemble Systems Generated through Sampling} \label{sec:unif_sample}

%
%
Here, we deal with a particular kind of stochastic ensemble system whose realizations are produced from samples of \eqref{eq:dlti}. This corresponds to the following scenario. Consider the individual components of the states to be nodes and the $\A$ matrix describing the interconnection between them. Suppose each node is aware of which neighbors are antagonistic and which neighbors are cooperative, but is unaware of the absolute magnitude of influence of its neighbors. The realizations of $\Ahat$ are hence obtained based on the relative order of influence of other nodes on a particular node.

In such a scenario, we assume that the stochastic ensemble system is produced in the following way. Suppose the vectors and matrices are transformed into probability mass functions over an underlying sample space. This mass function produces samples from the sample space which correspond to the realizations of the stochastic ensemble system. \mo{In particular}, we first define a function $\gamma : \real^n \to \real$, as
	$\gamma(\w) \ldef \one^\top \abs(\w)$,
\mo{and with a slight abuse of notation, we consider} $\gamma : \real^{n \times m} \to \real$ as
	$\gamma(\M) \ldef \max_{j \in \until{m}}\one^\top \abs(\M_{*j})$.
This helps us in describing \mo{a \emph{vector sampling scheme} procedure through the following lemma, whose proof omitted since it is trivial.}
\begin{lemma}\label{lem:unif_sample}
\thmtitle{Vector sampling scheme}
Consider a vector $\w \in \real^{{n}} \setminus \{\zero\}$. Then, let $ \wb \ldef (1/\gamma(\w))\abs(\w)$. Consider $\wb$ as a mass function on the sample space $\Omega = \{1,\cdots,{n}\}$ and let $\{s_1,\cdots, s_N\}$ be $N$ samples of $\Omega$. To each sample $s_i$, associate a vector $\svec^{(i)}$ as
{\small
\begin{align*}
[\svec^{(i)}]_j \ldef
\begin{cases}
	\gamma(\w), & \mathrm{if}\, j = s_i \,\,\mathrm{and} \,\, [\w]_j > 0;\\
	-\gamma(\w), & \mathrm{if}\, j = s_i \,\,\mathrm{and} \,\, [\w]_j < 0;\\
	0, & \mathrm{otherwise} \,.
\end{cases}
\end{align*}}
Then, the mean of the random variable corresponding to the $i\tth$ component of $\sum_{i = 1}^N \svec^{(i)}$ is $\w_i$. Finnally,
	$\lim_{N \to \infty} \mathbb{P}({ \|\w - \frac{1}{N} \textstyle{\sum}_{i = 1}^N \svec^{(i)}\| > \varepsilon}) \to 0,$
for each $\varepsilon > 0$.
\end{lemma}


Note that for matrices, the same procedure can be performed via a similar transformation by considering each row (or column) separately as vectors.

The stochastic ensemble system generated in this process has nice properties related to the system~\eqref{eq:dlti}. First, it retains the sparsity structure of the DLTI system~\eqref{eq:dlti}. Moreover, since each component is a Bernoulli random variable, using known properties we can provide convergence rates for each realization.  Next, we analyze the variance of the average control computed considering stochastic ensemble approximations produced as in Lemma~\ref{lem:unif_sample}. The next result follows from Hoeffding's inequality.

\begin{lemma}\thmtitle{Variance of average control}
\label{lem:variance}
Let $\mathcal{X}_0$, $\mathcal{X}_f$ and $\mathcal{A}$ be stochastic ensemble approximations of $\x_0$, $\x_f$ and $\A$ respectively generated using Lemma~\ref{lem:unif_sample}. Suppose $\{\uhat(\sigma,0)\}_{\sigma \in \integernonnegative}$ satisfy the hypothesis of Lemma~\ref{lem:single_step_ltv} and let $\v$ be as in~\eqref{eq:actual_control}. Then for each $\varepsilon > 0$,
\mo{
\begin{align}
\begin{array}{rl}
	\notag & \makesmall{\mathbb{P}({\| \B (\v - \frac{1}{N} \textstyle{\sum}_{\sigma = 1}^N\u^{(\smpind)}) \| > \varepsilon}) \leq 2C[ \exp ( - \frac{2 N \varepsilon^2}{9 \beta^2\gamma(\x_f)^2} )} \\
	\notag & \qquad \qquad \makesmall{ + \exp ( - \frac{2 N \varepsilon^2}{9\beta^2\gamma(\x_0)^2} ) + \exp ( - \frac{2 N \varepsilon^2}{9\beta^2\gamma(\A)^2} ) ]},
	\end{array}
\end{align}}
where $C$ is a constant dependent only on $n$.
\end{lemma}

\begin{proof}
First we characterize the convergence rates of the realizations of $\scriptxo, \scriptxf,$ and $\scriptA$ using the procedure in Lemma~\ref{lem:unif_sample}.
Note that 
\mo{
\begin{align}
	\label{eq:var_ineq} 
	\begin{array}{rl}
	& \makesmall{\| \x_0 \hspace{-.1cm}-\hspace{-.1cm} \frac{1}{N} \textstyle{\sum}_{s_0 = 1}^N \x_0^{(s_0)} \| \hspace{-.1cm} \leq \hspace{-.1cm} \eta(n) \max_{i} | [\x_0]_i \hspace{-.1cm}- \hspace{-.1cm}\frac{1}{N} \textstyle{\sum}_{s_0 = 1}^N [\x_0^{(s_0)} ]_i |},\\	
	& \makesmall{ \| \x_f \hspace{-.1cm}-\hspace{-.1cm} \frac{1}{N} \textstyle{\sum}_{s_f = 1}^N \x_f^{(s_f)} \| \hspace{-.1cm}\leq\hspace{-.1cm} \eta(n) \max_{i} | [\x_f]_i \hspace{-.1cm}-\hspace{-.1cm} \frac{1}{N} \textstyle{\sum}_{s_f = 1}^N [\x_f^{(s_f)} ]_i |},
	\end{array}
\end{align}}
where $\eta({n})$ is a constant dependent on $n$ such that $\|\w\| \leq \eta({n}) \|\w\|_\infty$, $\forall \w \,\in \real^{n}$. Moreover,
{\small
\begin{align}
\begin{array}{c}
	\label{eq:var_ineq3}  \makesmall{\| \A - \frac{1}{N} \sum_{s_A(1) = 1}^N \Ahat(s_A(1)) \| \leq}\\
	 \makesmall{\mu(n) \max\limits_{j\in \until{{n}}} \sum_{i\in \until{{n}}} | [\A]_{ij} \hspace{-.1cm}-\hspace{-.1cm} \frac{1}{N} \sum_{s_A(1) = 1}^N [\Ahat(s_A(1))]_{ij} |},
	\end{array}
\end{align}}
where $\mu({n})$ is a constant dependent on ${n}$ such that $\|\M\| \leq \mu({n}) \|\M\|_\infty$, $\forall \M \in \real^{{n} \times {n}}$.

Next, by using \mo{Hoeffding's inequality~\cite{WH:63}} and exploiting the Bernoulli nature of
the random variables, the right hand side of the inequalities in~\eqref{eq:var_ineq} and \eqref{eq:var_ineq3} can be bounded in probability for all $i \in \until{{n}}$ and for each $\varepsilon > 0$ as,
\mo{
\vspace{-.3cm}	
	\begin{align*}
	\begin{array}{rlll}
		& \makesmall{\mathbb{P}({ | [\x_0]_i \hspace{-.1cm}-\hspace{-.1cm} \frac{1}{N} \sum_{s_0 = 1}^N [\x_0^{(s_0)}]_i | \hspace{-.1cm}>\hspace{-.1cm} \varepsilon}) \leq 2\exp ( - \frac{2 N \varepsilon^2}{\gamma(\x_0)^2} )}\,,\\
		& \makesmall{\mathbb{P}({ | [\x_f]_i \hspace{-.1cm}-\hspace{-.1cm} \frac{1}{N} \sum_{s_f = 1}^N [\x_f^{(s_f)}]_i | \hspace{-.1cm}>\hspace{-.1cm} \varepsilon}) \hspace{-.1cm}\leq\hspace{-.1cm} 2\exp ( - \frac{2 N \varepsilon^2}{\gamma(\x_f)^2} )}\,,\\
		& \makesmall{\mathbb{P}({ | \A_{ij} \hspace{-.1cm}-\hspace{-.1cm} \frac{1}{N} \sum_{s_A(1) = 1}^N [\Ahat(s_A(1))]_{ij} | \hspace{-.1cm}>\hspace{-.1cm} \varepsilon})}
		\makesmall{ \leq \hspace{-.1cm}2\exp ( - \frac{2 N \varepsilon^2}{\gamma(\A)^2} )}
		\end{array}
	\end{align*}}
 \mo{Applying triangle inequality returns the results.}
\end{proof}

We can use the proof technique here to also characterize the rate of convergence in Lemma~\ref{lem:sparse_control_necessary} considering stochastic ensemble approximations produced as in Lemma~\ref{lem:unif_sample}.

\begin{lemma} \label{lem:final_state_variance}
\thmtitle{Variance of state trajectory}
Let $\mathcal{X}_0$, $\mathcal{A}$ and $\mathcal{A}$ be stochastic ensemble approximations of $\x_0$, $\A$ and $\B$ respectively, generated using Lemma~\ref{lem:unif_sample}. Then $\x(K)$ in Lemma \ref{lem:sparse_control_necessary} satisfies, for each $\varepsilon > 0$,

\vspace{-.4cm}
\mo{
{\small
\begin{align*}
	 \makesmall{\mathbb{P}({\| \x(K) - \frac{1}{N} \textstyle{\sum}_{\sigma = 1}^N\xhat(\sigma,K) \| > \varepsilon}) \leq 
	 2C(e^ {c_1} + e^{c_2}
	  + e^{c_3}),}
\end{align*}}
where $c_1=- \frac{2 N \varepsilon^2}{9\left\|\u \right\|^2\gamma(\B)^2}$, $c_2= - \frac{2 N \varepsilon^2}{9\left\|\A \right\|^2\gamma(x_0)^2}$, $c_3=- \frac{2 N \varepsilon^2}{9\left\|\x_0 \right\|^2\gamma(\A)^2}$ and $C$ only depends on $n$.}
\end{lemma}
\begin{proof}
The proof follows the same arguments as in the proofs of Lemmas~\ref{lem:sparse_control_necessary} and~\ref{lem:variance}.
\end{proof}
\vspace{-.2cm}
\subsection{Sample Averaging using Least Squares Error Minimization}

In the previous section, we assigned uniform weights to each realization in order to produce the stochastic ensemble approximation. For each sample run of the process, we can make this averaging method better by introducing a least square error minimization with the samples produced from the vector sampling scheme in Lemma~\ref{lem:unif_sample}. Since the weights associated with the optimization problem are restricted to be in the simplex $\smp^n$ (as per Definition~\ref{def:sparse_approx}), we approach this problem in two different ways. We describe each of them next and compare their performances later.

\paragraph{Accumulated Least Squares Error (ALSE) Minimization}
In this approach, we take into account all the realizations of the stochastic ensemble approximations $\scriptxo, \scriptxf,$ and $\scriptA$ all at once. Then we compute the associated averaging weights from least squares error minimization and finally use these weights for the computed control. Next, we provide a closed form solution to the least squares error minimization problem by exploiting the entries of the vectors in Lemma~\ref{lem:unif_sample}. 

The samples produced in Lemma~\ref{lem:unif_sample} have a very particular structure to them. In fact, they have a non-zero entry in one component and have \emph{zeros} everywhere else (\emph{i.e.} their support is a singleton set). Hence, the whole problem boils down to adjusting the weights individually for the components in order to minimize the norm distance of the weighted sum of the samples and the vector. We provide this in the next result.

\begin{lemma}\label{lem:lsq_sol}
\thmtitle{Least square error problem solution for singleton support samples}
Consider a vector $\w \in \real^{{n}} \setminus \{\zero\}$. Let $\wb \ldef (1/\gamma(\w))\abs(\w)$ and suppose $\{\svec^{(1)},\cdots, \svec^{(N)}\}$ is a set of $N$ sample vectors produced using the sampling scheme in Lemma~\ref{lem:unif_sample}. Let $\S$ be the matrix whose $i\tth$ column is $\svec^{(i)}$. Then a solution to

\vspace{-.4cm}
{\small
\begin{align}
	\min_{\y \in \smp^N} \,\, \| \S\y - \w \|_2^2\,,
	\label{eq:lsq}
\end{align}}
is given by $\y^*$ with $[\y^*]_j = \alpha_j/n_j$, where $\alpha_j = |[\w]_k|$ with $k$ such that $[\svec^{(j)}]_k \neq 0$ and $n_j = |\{i \in \until{N} \,|\, [\svec^{(i)}]_k \neq 0\}|$.
\end{lemma}

\nm{
\begin{proof}
Note that the problem is convex in $\y$ and that the set $\until{{n}}$ can be split into two disjoint sets $\mathcal{D}_1 \ldef \{p \in \until{{n}} \,|\, \exists \, q \,\mathrm{s.t.}\, [\svec^{(q)}]_p \neq 0\}$ and $\mathcal{D}_2 \ldef \until{{n}} \setminus \mathcal{D}_1$. Now, since the support of $\svec^{(i)}$ is a singleton, the problem can be rewritten as

\vspace{-.5cm}
{\small
\begin{align*}
	\min_{\y \in \smp^N} \,\,  \textstyle{\sum}_{i \in \mathcal{D}_1} ( \textstyle{\sum}_{j \in \smpind(s)} [\svec^{(j)}]_i\y_{j}  - \w_i)^2 + \textstyle{\sum}_{i \in \mathcal{D}_2} \w_i^2 \rdef J(\y) \,,
\end{align*}}
where $\smpind(i) \ldef \{p \in \until{N} \,|\, [\svec^{(p)}]_i \neq 0\}$. It is easy to see that for the solution $\y^*$ in the claim, the value of the cost function is $J(\y^*) = \sum_{i \in \mathcal{D}_2} \w_i^2$. As this is the minimum value of $J(\cdot)$, the proof is complete.
\end{proof}

It is worth noting that if the set $\mathcal{D}_2$ in the proof of Lemma \ref{lem:lsq_sol} becomes empty, then $J(\y^*) = 0$. This means that if the samples are rich enough so that all the non-zero entries of the original vector appear at least once in the set $\{\svec^{(1)},\cdots\hspace{-.1cm}, \svec^{(N)}\}$, then the weights assigned using the least squares problem recreates the original vector perfectly. This does not translate to the case for matrices as then the associated weights to the matrix samples are coupled using multiple non-zero entries across the matrix, as illustrated in Section~\ref{sec:sims}.
}

\paragraph{Streaming Least Squares Error (SLSE) Minimization} 
In this case, the controls are attained from sequentially changing the realizations of the vectors and matrices. For the uniform weights case, the solution can be made better by adding more samples. For the ALSE minimization case, the weights have to be determined beforehand and then the controls computed, since we do not assign the same weights to each sample. Moreover, as the number of samples increases, the weights assigned to the previous samples need to be recomputed in order to satisfy~\eqref{eq:lsq}. 

However for the SLSE minimization case, the weights to the controls can be updated sequentially in order to construct a sub-optimal solution to~\eqref{eq:lsq}. Essentially, suppose $\E_1$ and $\E_2$ are estimates of $\E^*$, with $\E_1, \E_2$ and $\E^*$ being comparable entities from $\{\Mcal{X}_0, \Mcal{X}_f, \Mcal{A}\}$, \emph{i.e.} they are arbitrary elements in the same sequence. Then we solve for
{\small
\begin{align}
	\min_{(w_1,w_2) \in \smp^2} & \|w_1 \E_1 + w_2 \E_2 - \E^*\|_2^2\,.
	\label{eq:streaming_opt}
\end{align}}
The computed controls are updated using the same weights. We explain this in Algorithm \ref{algo:streaming}.
{\small
\begin{algorithm}{Sequential control computation with least square error estimation}
\begin{algorithmic}[1]
\REQUIRE $\mathcal{X}_0$, $\mathcal{X}_f$, $\mathcal{A}$, $\x_0$, $\x_f$, $\A$,$\B$, \texttt{errorBound}
\ENSURE control $\v(k)$
\STATE \texttt{error} $\gets \infty$; \quad $\xhat_0 \gets \xhat(\sigma,0) \in \Mcal{X}_0$;
\STATE  $\xhat_f \gets \x_f(s_f(\sigma)) \in \Mcal{X}_f$; \quad $\Ahat \gets \Ahat(s_A(\sigma,1)) \in \Mcal{A}$; 
\STATE $\v \gets \widehat{\v}$ such that $\xhat_f = \Ahat\xhat_0 + \B \widehat{\v}$
\WHILE{\texttt{error} $>$ \texttt{errorBound}}
\STATE Randomly choose $i \in \{1,2,3\}$
\IF{$i = 1$}
	\STATE update $\xhat_0$
	\STATE $(w_1^*,w_2^*) \gets$ solution of~\eqref{eq:streaming_opt} with $\E_1 = $ old $\xhat_0$, $\E_2 = $ new $\xhat_0$ and $\E^* = \x_0$. 
\ELSIF{$i = 2$}
	\STATE update $\xhat_f$
	\STATE $(w_1^*,w_2^*) \gets$ solution of~\eqref{eq:streaming_opt} with $\E_1 = $ old $\xhat_f$, $\E_2 = $ new $\xhat_f$ and $\E^* = \x_f$.
\ELSIF{$i = 3$}
	\STATE update $\Ahat$
	\STATE $(w_1^*,w_2^*) \gets$ solution of~\eqref{eq:streaming_opt} with $\E_1 = $ old $\Ahat$, $\E_2 = $ new $\Ahat$ and $\E^* = \A$.
\ENDIF
\STATE $\v' \gets \widehat{\v}$ such that $\xhat_f = \Ahat\xhat_0 + \B \widehat{\v}$
\STATE \texttt{error} $\gets$ $\|(w_1^*-1)\v + w_2^*\v'\|$
\STATE $\v \gets w_1^*\v + w_2^*\v'$
\ENDWHILE
\end{algorithmic}
\label{algo:streaming}
\end{algorithm}}

We conclude this section by comparing the averaging methods provided earlier.
\vspace{-.2cm}
\subsection{Comparison between the Averaging Methods}
Since the weights obtained in the least squares error minimization are obtained from an optimization problem, we can compare their convergence rate with the uniform averaging scheme. 

\begin{lemma}\label{lem:compare}
\thmtitle{Comparison between averaging methods}
Suppose $\{\E_i\}_{i=1}^N$ are estimates of $\E^*$, with $\{\E_i\}_{i=1}^N$ and $\E^*$ being comparable entities from $\{\Mcal{X}_0, \Mcal{X}_f, \Mcal{A}\}$. Let
$\F_{\mathrm{uniform}} \ldef (1/N)\sum_{i=1}^N \E_i$. Let $\F_{\mathrm{ALSE}} \ldef \sum_{i=1}^N w^*_i\E_i$ with $[w^*_1,\cdots,w^*_N]$ being a solution of

\vspace{-.4cm}
{\small
\begin{align}
	 \min_{[w_1,\cdots,w_N] \in \smp^N} \|\E^* - \textstyle{\sum}_{i=1}^N w_i\E_i \|_2^2\,.
	 \label{eq:lsq_general}
\end{align}}
Finally, let $\F_{\mathrm{SLSE}} \ldef \F(N)$ where $\F(N)$ is a solution to the difference equation $\F(n) = v^*_1 \F(n-1) + v^*_2\F(n-2)$ starting from $\F(1) = \E_1, \F(2) = \E_2$ and where $v^*_i$ comes from Algorithm \ref{algo:streaming}. Then,
{\small
\begin{align}
	\|\E^* \hspace{-.1cm}- \hspace{-.1cm}\F_{\mathrm{ALSE}}\| \hspace{-.1cm}\leq\hspace{-.1cm} \|\E^*\hspace{-.1cm} - \hspace{-.1cm}\F_{\mathrm{SLSE}}\| \hspace{-.1cm}\leq\hspace{-.1cm} \|\E^* \hspace{-.1cm}-\hspace{-.1cm} \F_{\mathrm{uniform}}\|\,. 
	\label{eq:estimate_error_order}
\end{align}}
\end{lemma} 

\begin{proof}
Note that $[v^*_1,\cdots,v^*_N]$ is a feasible solution to~\eqref{eq:lsq_general}. \mo{This results in the first inequality. The second inequality holds since $(1/N)\one$ is a feasible solution to the optimization problem in Algorithm \ref{algo:streaming}.}
\end{proof}
\section{Simulations and Analysis}\label{sec:sims}
In this section, we provide simulation examples to validate our results in two scenarios. In the first case, we verify Lemma~\ref{lem:compare} in a power systems example. In the second case, we study the error in the computed state for a system with large dimensions. \nm{We use CVX MATLAB toolbox~\cite{MG-SB:14} to solve the optimization problems.}
\subsection{Error in Computed Control} \label{sec:power_example} 
\vspace{-.4cm}
Here, we take an example system \mo{in the form of \eqref{eq:dlti}} from~\cite{FP-FD-FB:12a} using MATLAB's power toolbox. The system dimensions are $n = 10$ and $m = 5$ (with $\A \in \real^{n \times n}$, $\B \in \real^{n \times m}$) and we choose the initial condition randomly. 

\begin{figure}[h!]
	\begin{center}
	\begin{tabular}{cc}
		\includegraphics[trim = 1.3in 3.1in 1.5in 3.3in, clip, scale=0.28]{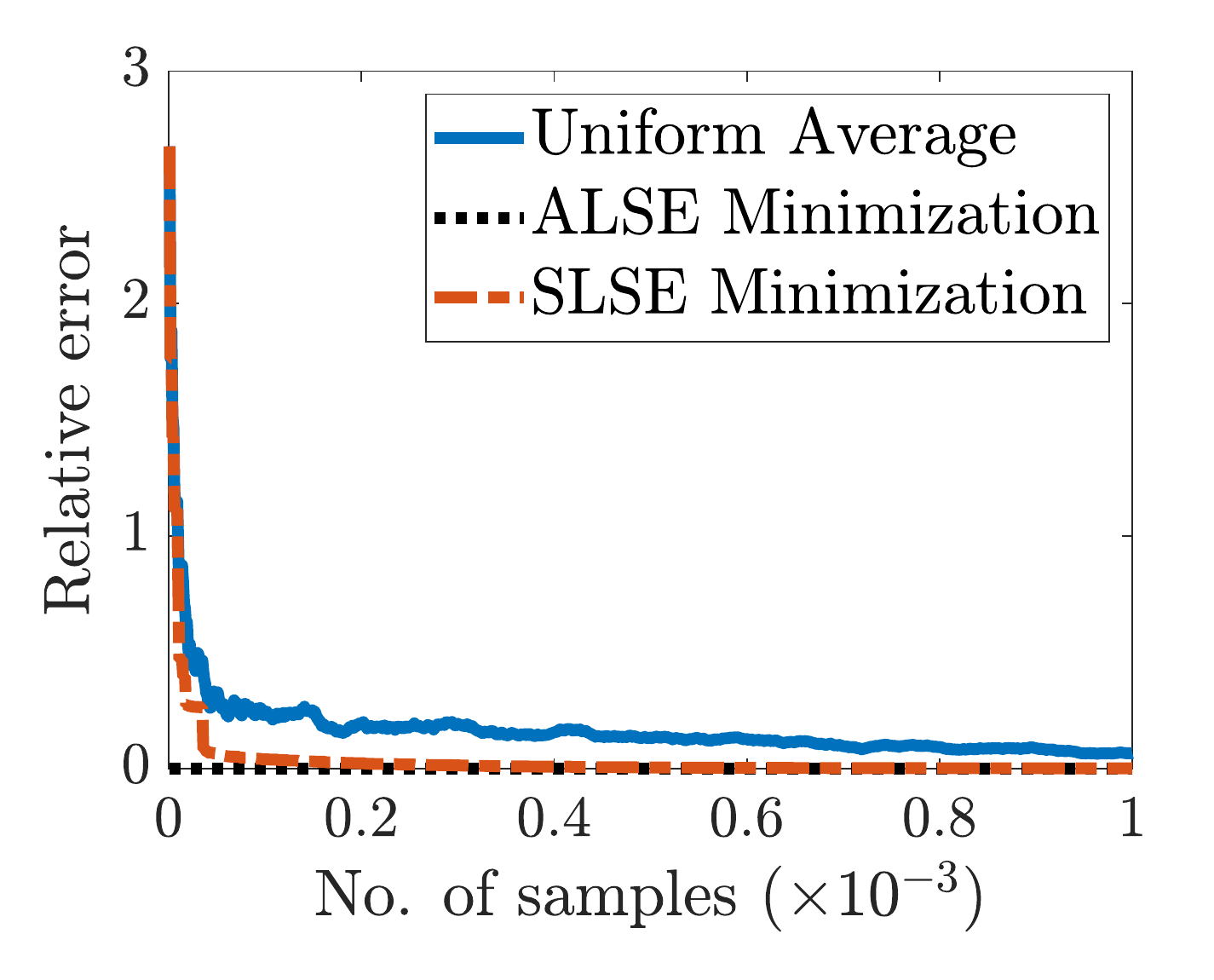} & \includegraphics[trim = 1.3in 3.1in 1.5in 3.3in, clip, scale=0.28]{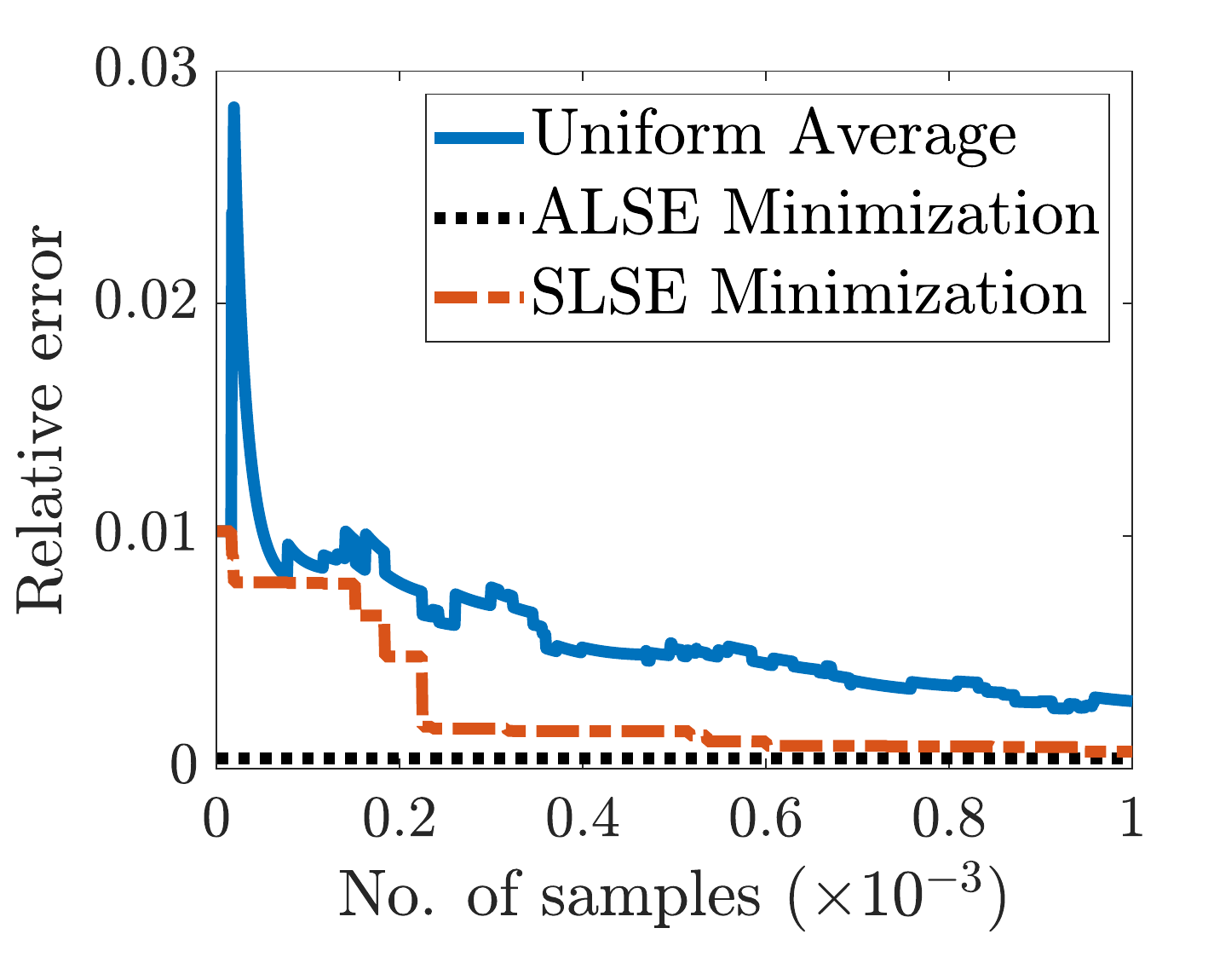}\\
		\hypertarget{subfig:error_vector}{(a)} & \hypertarget{subfig:error_matrix}{(b)}
	\end{tabular}
	\vspace{-.4cm}
		\caption{{\small Relative error in estimation versus number of samples for example in Section \ref{sec:power_example}. (a) Relative error in estimated initial condition $\x_0$. (b) Relative error in estimated system matrix $\A$.}}
		\label{fig:estimation_error}
	\end{center}
\end{figure}
\vspace{-.5cm}
\nm{First, we verify the ordering in estimation errors for the different averaging methods in Figure~\ref{fig:estimation_error}. This is in accordance with~\eqref{eq:estimate_error_order}. Moreover, note that in Figure~\ref{fig:estimation_error}\hyperlink{subfig:error_vector}{(a)}, it is possible to obtain \emph{zero} relative error for the initial condition estimation using ALSE Minimization and SLSE Minimization as per the discussion following Lemma~\ref{lem:lsq_sol}.} 
\nm{Next, we deal with the errors in the computed control and the computed state.}The goal is to drive the state close to the origin in one time step. The actual control that does this was computed by taking $\u_\mathrm{actual} = \B^{\dagger}(-\A\x_0)$, where $\B^{\dagger}$ denotes the Moore-Penrose pseudoinverse. Then the three different controls $\u_\mathrm{computed}$ were computed using the methods listed in this paper (\emph{i.e.} uniform average, ALSE minimization and SLSE minimization). The comparison of the relative error in computed control with respect to the number of samples is shown in Figure \ref{fig:power_sys}\hyperlink{subfig:power_control}{(a)}. Note that the $\u_\mathrm{computed}$ for the ALSE minimization case was computed after considering all the samples. We also computed the relative error in the computed state with respect to the number of samples and present the results in in Figure \ref{fig:power_sys}\hyperlink{subfig:power_state}{(b)}. \nm{It is worth mentioning that the error in the computed control and the computed states do not always follow the same ordering as the estimation errors of the initial condition and the system matrix}. 
%

\begin{figure}[h!]
	\begin{center}
	\begin{tabular}{cc}
		\includegraphics[trim = 1.3in 3.1in 1.5in 3.3in, clip, scale=0.28]{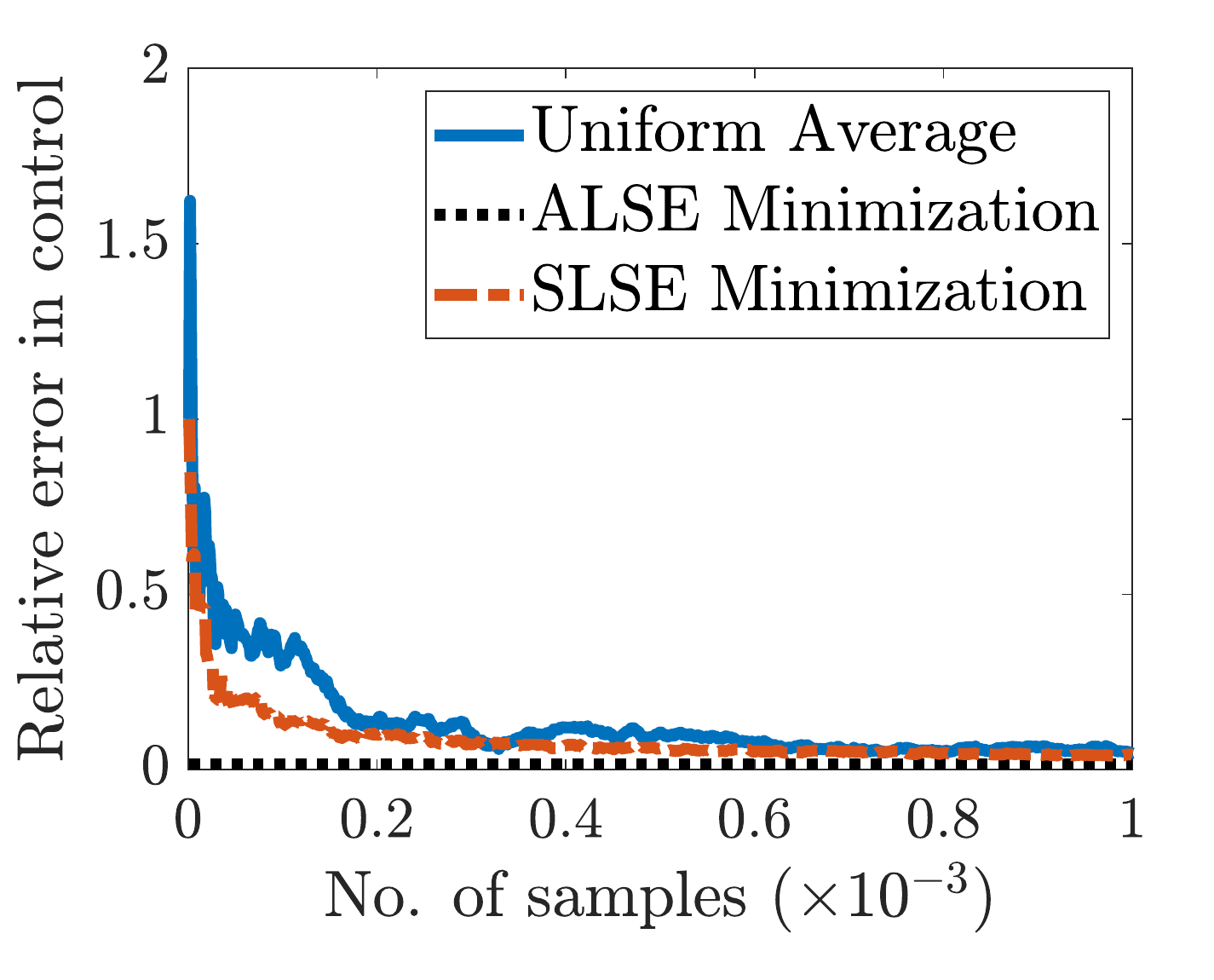} & \includegraphics[trim = 1.3in 3.1in 1.5in 3.3in, clip, scale=0.28]{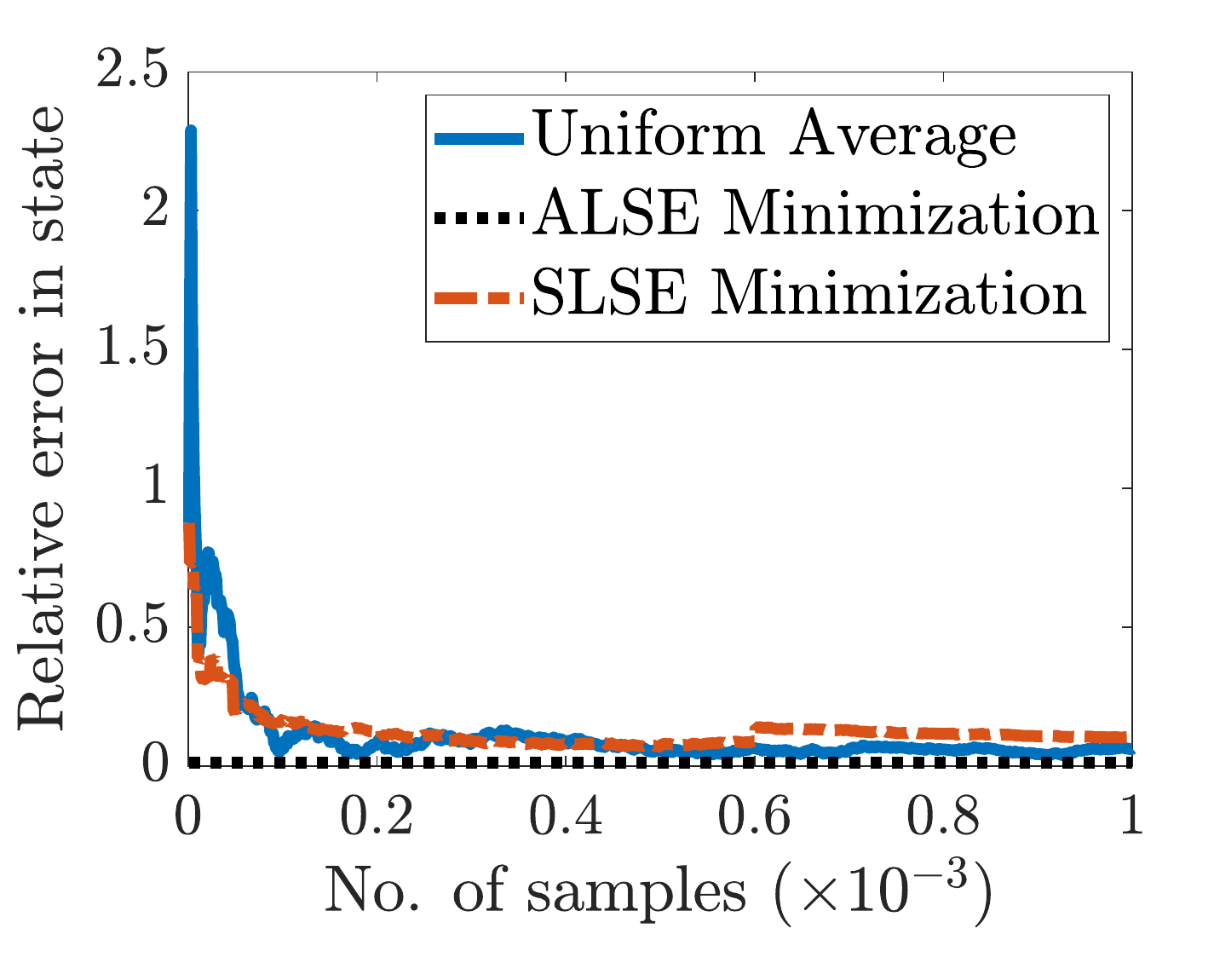}\\
		\hypertarget{subfig:power_control}{(a)} & \hypertarget{subfig:power_state}{(b)}
	\end{tabular}
	\vspace{-.3cm}
		\caption{{\small Relative error versus number of samples for example in Section \ref{sec:power_example}. (a) Relative error in control $\frac{\|\u_\mathrm{computed} - \u_\mathrm{actual}\|}{\|\u_\mathrm{actual}\|}$. (b) Relative error in state.}}
		\label{fig:power_sys}
	\end{center}
\end{figure}
\vspace{-.8cm}
\subsection{Error in Computed State} \label{sec:large_example}
Here, we simulate our algorithms in a trajectory tracking problem for a system with large dimensions. Notice that the sample averaging method works for one time step controls only. Thus, it can be used to track a predefined trajectory of states. For this example, we take the system dimensions as $n = 100$, and $m = 80$ and generate $\A \in \real^{n \times m},\B \in \real^{n \times m},\x(0) \in \real^{n},\x(1) \in \real^{n},$ and $\x(2) \in \real^{n}$ randomly. The matrices $\A$, and $\B$ are used as the system matrices and $\x(0)$ is used as the initial condition. The states $\x(1)$, and $\x(2)$ are used as the reference trajectory.

Similar to the previous case, for each $t \in \{0,1\}$, the control for each realization $\A^{(s)}, \x^{(s)}(t),$ and $\x^{(s)}(t+1)$ is computed as $\u_\mathrm{computed} = \B^\dagger(\x^{(s)}(t+1) - \A^{(s)} \x^{(s)}(t))$. The relative error in the computed states with respect to the number of samples is given in Figure~\ref{fig:traj_track}. 
Notice that (similar to what we observed in Section~\ref{sec:power_example}) the ordering given in Lemma~\ref{lem:compare} that holds \nm{for the estimation errors for the initial condition and the system matrix},  may not hold for the computed states as shown in Figure~\ref{fig:traj_track}. Also, the errors in Figure~\ref{fig:traj_track} are much larger than the errors in Figure~\ref{fig:power_sys}. This is because, the dimensions in this case are much larger than the previous case. Using the idea of convergence of the computed control in Lemma~\ref{lem:single_step_ltv}, we can give similar results for the convergence of computed states, omitted due to the lack of space.
\begin{figure}[h!]
	\begin{center}
	\begin{tabular}{cc}
		\includegraphics[trim = 1.3in 3.1in 1.5in 3.3in, clip, scale=0.28]{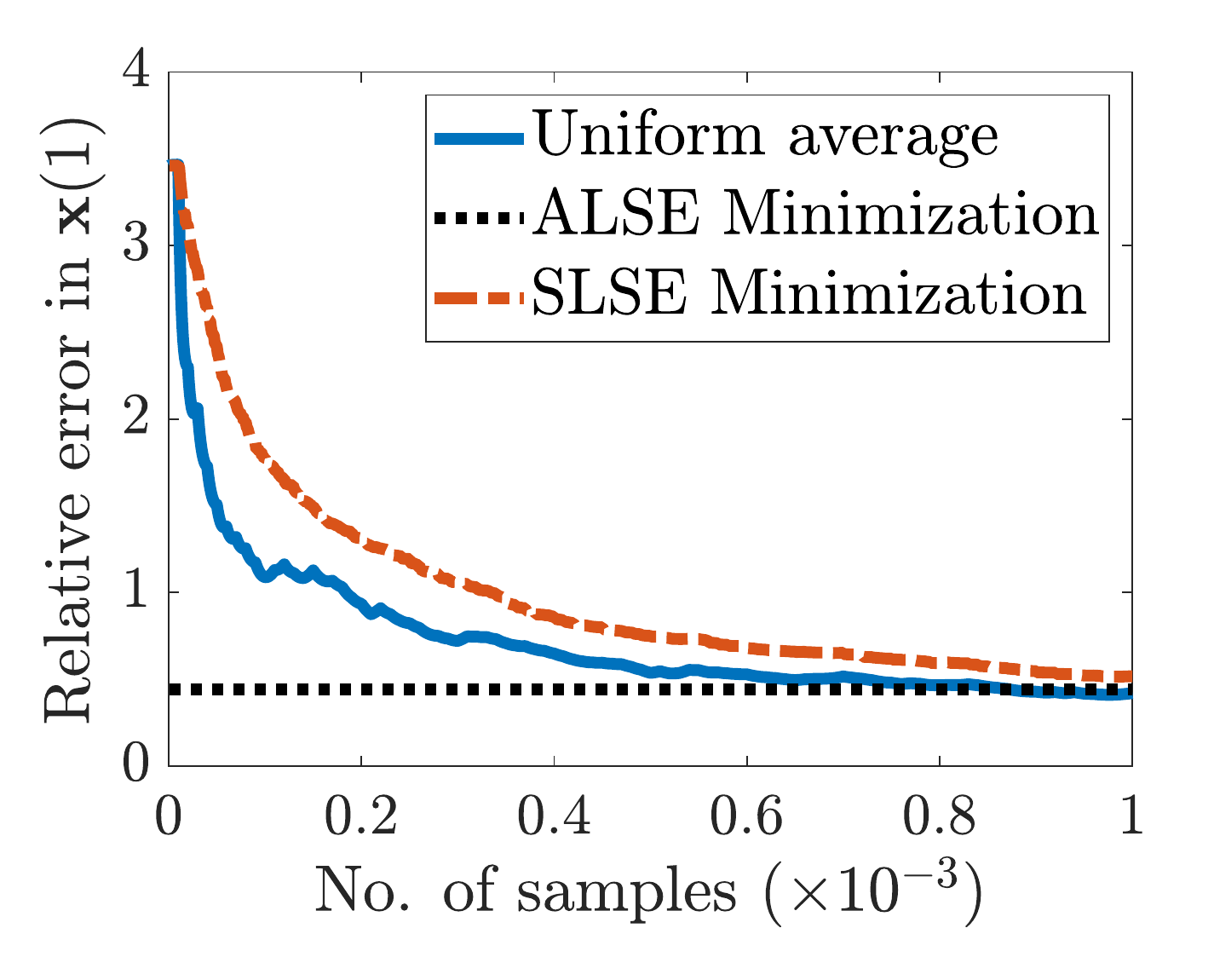} & \includegraphics[trim = 1.3in 3.1in 1.5in 3.3in, clip, scale=0.28]{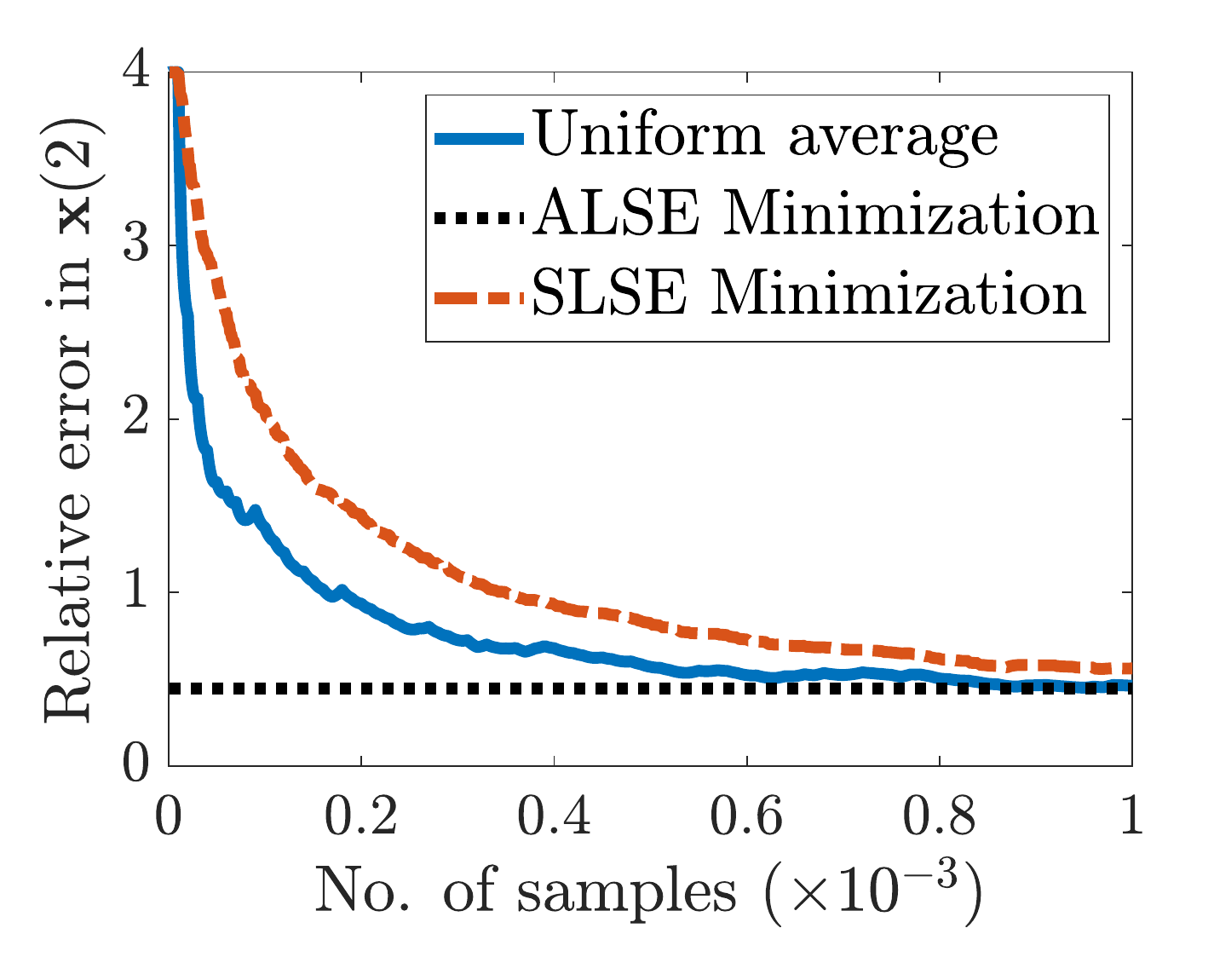}\\
		(a) & (b)
	\end{tabular}
	\vspace{-.4cm}
		\caption{{\small Relative error in state $\frac{\|\x_\mathrm{computed} - \x_\mathrm{actual}\|}{\|\x_\mathrm{actual}\|}$ versus number of samples for example in Section \ref{sec:large_example}. (a) Relative error in $\x(1)$. (b) Relative error in $\x(2)$.}}
		\label{fig:traj_track}
	\end{center}
\end{figure}
\section{Conclusion}
\mo{In this paper, we studied the reachability of a new class of stochastic ensemble systems, where are not required to have a continuum of ensembles.} We provided a necessary condition for reachability under multiple time steps for the limiting DLTI system using approximate reachability for the stochastic ensemble system. We also provided a sufficient condition that averages the control obtained from the sample reachability condition for the ensemble system to find a control for the DLTI system. We provided convergence rates when the stochastic ensemble systems are produced using sampling.
In the future, \mo{we will} extend the sufficient condition to multi-time and continuous-time cases, and will study stochastic ensemble approximations of nonlinear systems. 

\hspace{-.5cm}
\tiny
\bibliographystyle{IEEEtran}
\bibliography{ref.bib}

\end{document}